\documentclass[conference]{IEEEtran}
\usepackage{amsmath,amssymb,amsfonts, amsthm}
\usepackage{mathtools}
\usepackage{dsfont}
\usepackage{algorithmic}
\usepackage{algorithm2e}
\usepackage{float} 
\usepackage{bm}
\usepackage{afterpage}
\usepackage{graphicx}
\usepackage{textcomp}
\usepackage{xcolor}
\usepackage{comment}
\usepackage{gensymb}
\usepackage{float}
\usepackage{import}
\usepackage{pgfplots}
\usepackage{pgfplotstable}
\usepgfplotslibrary{fillbetween}
\usepackage[normalem]{ulem}
\usepackage{mdefs}
\usepackage{cite}
\usepackage{lipsum}
\usepackage{subcaption}
\usepackage[utf8]{inputenc}
\usepackage{hyperref}
\DeclareUnicodeCharacter{2061}{}
\usepackage[top=0.72in, bottom=1.0in, left=0.63in, right=0.66in]{geometry} 

\usepackage{algorithmic}
\usepackage{algorithm2e}

\floatstyle{ruled}
\newfloat{algorithm}{tbp}{loa}
\providecommand{\algorithmname}{Algorithm}
\floatname{algorithm}{\protect\algorithmname}

\def\BibTeX{{\rm B\kern-.05em{\sc i\kern-.025em b}\kern-.08em T\kern-.1667em\lower.7ex\hbox{E}\kern-.125emX}}

\def\argmax{\operatornamewithlimits{arg\,max}}

\def\R{\text{R}}
\def\E{\text{E}}

\newtheorem{prop}{Proposition}[section]

\newcommand{\etc}{\textit{etc. }}

\title{Analyzing and Enhancing Queue Sampling for Energy-Efficient Remote Control of Bandits}

\author{\IEEEauthorblockN{Hiba Dakdouk, Mohamed Sana, Mattia Merluzzi}
\IEEEauthorblockA{CEA-Leti, Université Grenoble Alpes, F-38000 Grenoble, France\\
Email: \{hiba.dakdouk, mohamed.sana, mattia.merluzzi\}@cea.fr}}

\newcommand{\titleheader}{This work has been accepted for publication in 2024 IEEE International Mediterranean Conference on Communications and Networking}

\def\ps@IEEEtitlepagestyle{%
\def\@oddhead{\mbox{}\scriptsize \titleheader \rightmark \hfil }%
}
\makeatother

\begin{document}

\maketitle

\begin{abstract}
In recent years, the integration of communication and control systems has gained significant traction in various domains, ranging from autonomous vehicles to industrial automation and beyond.
Multi-armed bandit (MAB) algorithms have proven their effectiveness as a robust framework for solving control problems.
In this work, we investigate the use of MAB algorithms to control remote devices, which faces considerable challenges primarily represented by latency and reliability.
We analyze the effectiveness of MABs operating in environments where the action feedback from controlled devices is transmitted over an unreliable communication channel and stored in a Geo/Geo/1 queue.
We investigate the impact of queue sampling strategies on the MAB performance, and introduce a new stochastic approach.
Its performance in terms of regret is evaluated against established algorithms in the literature for both upper confidence bound (UCB) and Thompson Sampling (TS) algorithms.
Additionally, we study the trade-off between maximizing rewards and minimizing energy consumption.

\end{abstract}
\begin{IEEEkeywords}
Multi-armed Bandits, Joint Communication and Control, Queue Sampling, Latency and Reliability
\end{IEEEkeywords}
\section{Introduction}

In the era of Industry 4.0, the primary goal is to reduce human involvement within industrial processes \cite{strinati20196g}.
This requires a tight integration of automatic control systems and communication technologies in order to support various  application areas, such as  intelligent traffic control systems and automatic warehouse management systems\cite{bemporad2010networked}.
One significant challenge in such systems is the efficient control of remote devices to achieve desired objectives in dynamic and uncertain environments. 
In these scenarios, decisions must be taken in real-time and inferred from limited information (generated by the device itself or from the surrounding environment), while also considering the trade-offs between exploration and exploitation of different control strategies.

Multi-armed bandit (MAB) algorithms offer a powerful framework for addressing control problems in these settings. 
Originally developed in the context of sequential decision-making under uncertainty, MAB algorithms were designed to navigate the exploration-exploitation dilemma encountered by single agents, \ien, the same agent selects actions and interacts with the environment.
In the context of joint communication and control, MAB algorithms provide a principled approach to control remote devices/systems.
Unlike traditional setups, in remote MAB-controlled systems, the agent serving as a decision-maker is decoupled from the device that interacts with the environment.
This applies to various applications including autonomous vehicles, robotics, distributed computing systems, \etc
However, in such settings, latency and reliability stand out as major challenges that are not addressed in traditional MAB algorithms.
The potential transmissions interruptions coupled with the time gap between applying an action and observing its corresponding feedback, can significantly mitigate the system effectiveness.
Moreover, the energy consumption naturally demands high attention in any sustainable development effort, and/or in the presence of resource poor devices\cite{nota2020energy}.  
Therefore, any novel approach  must prioritize energy efficiency. 

In this work, we focus on the adaptation of standard MAB algorithms to such challenging settings.
We consider a decision making agent controlling the actions of a remote device.
The agent transmits the action to the controlled device, and the latter sends back the environment feedback over unreliable channel causing potential losses as shown in Fig.\ref{fig:system_model}.
If received successfully, it will be stored in a queue until being served by the agent, causing potential delays.

\textbf{Related Work: } 
The remote control of bandits has been investigated in prior studies.
Motivated by the web advertisement applications, the work in \cite{joulani2013online} studies the online learning problems with randomly delayed feedback.
The authors provide general algorithms that transform standard algorithms devised for the non-delayed case, termed as \textit{BASE}, into ones that handle delayed feedback. 
They show that the delay inflates the regret in a multiplicative way in adversarial problems, while it is only additive in stochastic problems.
For stochastic environments, they propose the QPM-D algorithm, which stores the feedback information coming from the environment in separate queues for each arm, enabling interacting with the queues instead of the environment.
Inspired by such a queue method, the authors in \cite{mandel2015queue} propose the stochastic delayed bandit (SDB) algorithm that empirically outperforms QPM-D while preserving its theoretical guarantees. 
It combines the BASE algorithm with a heuristic policy., which jointly guide actions in the environment. 
The BASE algorithm is updated by observing rewards from queues, while the heuristic is updated directly from the environment.
Following the same queue methodology, the work in \cite{gupta2023remote} investigates empirically the regret performance of the remote control of standard MABs. 
The authors model the communication network as a \textit{Geo/Geo/1} queue.
Motivated by the work in \cite{joulani2013online,mandel2015queue}, they propose heuristic policies that use Thompson Sampling algorithm for arm selection and reward queues for obtaining reward feedback.
They show through simulations that these heuristics have better regret performance compared to SDB algorithm from \cite{mandel2015queue}. 

In another framework, the authors in  \cite{desautels2014parallelizing}, analyze the  Gaussian Process Bandits encountering delayed rewards. 
Under the assumption that the delay follows a known stochastic distribution, in \cite{vernade2017stochastic}, the authors develop two UCB-based algorithms with delay-corrected indices.
This work is extended in \cite{vernade2020linear} to the linear contextual bandits.
The authors study the delayed noisy rewards with no prior information on the delay distribution, and develop a new algorithm based on Linear UCB (LinUCB) algorithm. 
The remote control  of contextual bandits has also been investigated in \cite{dudik2011efficient, bouneffouf2022linear, pase2022remote}.
The work in \cite{dudik2011efficient} assumes that the delay is a fixed, known, and constant, and proves it has an additive effect on the regret.
In \cite{bouneffouf2022linear}, the authors study the problem of missing rewards. 
They build upon LinUCB and use unsupervised learning to estimate the reward from the observed context.
In \cite{pase2022remote}, the authors study the problem of sending the policy through a rate-limited communication channel. 

In contrast to the queue-based methods, in this work we do not assume any storage capability at the agent side.
We rather focus on exploring methods to sample packets from the server queue such that to reduce the effect of delayed or lost packets.
We introduce a novel, simple sampling strategy that surpasses traditional sampling methods such as first-in-first-out (FIFO) and last-in-first-out (LIFO).
The performance of our approach is evaluated against existing algorithms in the literature by  investigating the trade-off between regret minimization and energy consumption via numerical analysis.

\section{System Model and Problem Formulation}

We consider an actuator\footnote{Although we refer to it as ``actuator", it could represent any device or system responsible for taking actions and interacting with the environment} interacting with the environment by taking actions and receiving corresponding feedback/reward.
The environment is modeled as a $K$-armed bandit with independent and identically distributed (\iidn) rewards.
Let $\mathcal{K}$ be the set of $K$ actions (arms) available to the actuator.
We consider a stationary environment, such that the rewards of each arm are generated from a stochastic distribution $\mathcal{X}$ that is stationary, \ien, it does not change over time.
As presented in Fig. \ref{fig:system_model}, at each time slot $t$, the actuator applies one arm $a_t \in \mathcal{K}$ to the environment, and observes its corresponding reward $r_{a_t,t}\sim \mathcal{X}(\theta_{a_t})  $, where $\theta_a$ is the mean reward of arm $a$.
Inspired by \cite{gupta2023remote}, we assume the actuator is controlled by a remote agent (the decision maker) that sends  the action to be taken by the actuator at each time slot, over a reliable and zero-delay channel\footnote{This assumption can also be relaxed to incorporate packet losses and delays on this communication link.}.
After each interaction with the environment, the latter then sends back a packet containing the selected arm and the observed reward at time $t$, \ien, $(a_t,r_t)$, over an unreliable channel causing potential losses.
The packet is successfully transmitted over the communication channel and admitted to a \textit{Geo/Geo/1} queue located at the server, with a probability $\lambda \in [0,1]$; the admission process is \iid over time.
At the end of every time slot, a packet from the queue is then served and observed by the agent with a probability $\mu \in [0,1]$\footnote{The agent and the server are depicted as separate entities in the diagram for readability, although they can be physically co-located.}; the service process is also assumed to be \iid over time.
Consequently, the agent cannot observe all the rewards due to losses when $\lambda<1$, and the observations will be delayed when $\mu<1$.
The parameters $\lambda$ and $\mu$ are indicative of the agent's communication and computational resources respectively.

\begin{figure}[!t]
    \centering
    \includegraphics[width=0.85\linewidth]{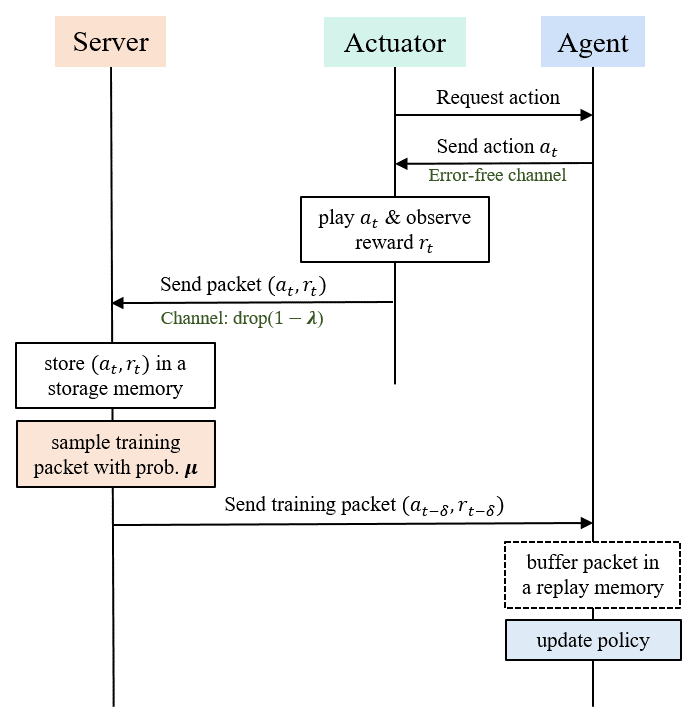}
    
    \caption{System model of remotely controlled system.}
    \label{fig:system_model}
\end{figure}

The packet to be observed by the agent is sampled from the queue based on a predefined process. 
Let $\pi$ denote the sampling strategy followed by the agent to pick packets from the queue, \ien, $\pi(t)=(a,r)$.
Let $L(t)$ be the length of the queue at time $t$ and  $s$ denote the index of a packet in the queue.
Then $(a,r)_s$ corresponds to the packet of index $s$.
Conventional examples of $\pi$ include: $\pi_{\texttt{FIFO}}(t)=(a,r)_1$ and $\pi_{\texttt{LIFO}}(t)=(a,r)_{L(t)}$.

As in standard MABs, the agent's goal is to maximize the cumulative rewards over time; $\max \sum_{t=0}^{T} r_t$.
Consequently, it follows an algorithm $\mathcal{A}$ to select the arms at each time slot, \ien, $a_t \sim \mathcal{A}$. 
The agent updates algorithm $\mathcal{A}$ whenever it observes a new packet from the queue. 
Algorithm $\mathcal{A}$ can be any commonly used bandit algorithm, including $\epsilon-$greedy, UCB and TS algorithms.
Algorithm \ref{alg:network-queue} shows how standard bandit algorithms are adjusted to suit the queue-based remote controlled model, where $\mathcal{B}$ refers to the Bernoulli distribution.
If the agent is equipped with storage capability as in \cite{joulani2013online, mandel2015queue, gupta2023remote}, a replay memory could be introduced for the buffering and reuse of observed packets/rewards.
Consequently, the agent can undergo a seamless and continuous training due to the environment's stationarity, as the training of algorithm $\mathcal{A}$ no longer depends on the service rate, but rather it relies on the content stored within the replay buffer.
In this work, we focus on scenarios where the agent does not possess storage capabilities.

\begin{algorithm}[htb]
  \caption{Queue-based Remote Controlled MAB  (QR-MAB)}
  \label{alg:network-queue}
    {\bf Inputs:} $\mathcal{K}$, $T$, $\mathcal{A}$, $\pi$
\begin {algorithmic}[1]
        \FOR{$t=1$ to $T$}
           \STATE Draw $a_t \sim \mathcal{A}$
           \STATE Observe $r_t \sim \mathcal{X}(\theta_{a_t})$
           \IF{$\mathcal{B}(\lambda)$=1}
                \STATE add $(a_t,r_t)$ to the queue
           \ENDIF
           \IF{$\mathcal{B}(\mu)$=1}
                \STATE sample a packet $\pi(t)=(a,r)$
                \STATE remove $(a,r)$ from queue
                \STATE update $\mathcal{A}$ with $(a,r)$
            \ENDIF
        \ENDFOR
\end {algorithmic}
\end{algorithm}

The decision of $\mathcal{A}$ at time $t$ is then dependent on the sequence of feedback acquired from the network up to time $t-1$ which in turn relies on the sampling policy $\pi$ employed.
For a time horizon $T$, the cumulative pseudo-regret of a policy $\mathcal{A}$ is defined as:
\begin{equation}
\label{eq:pseudo-regret}
    \mathcal{R}^\mathcal{A}(T)=T\theta^*-\sum_{t=1}^T \theta_{a_t},
\end{equation}
where, $\theta^*=\max_{a}\theta_a$ is the best mean reward.

As shown in \cite{gupta2023remote}, the performance of any algorithm is influenced by the sampling policy $\pi$ followed by the agent, and obviously on the system parameters $\lambda$ and $\mu$. 
Indeed, each algorithm is impacted differently by the same sampling policy.
In this work, we aim to find the optimal sampling policy for a given MAB algorithm $\mathcal{A}$ that minimizes the cumulative pseudo-regret, \ien, equivalently maximizing the total acquired rewards.

\section{Queue-based Remote Control of Upper Confidence Bound Algorithm}

In this work, we focus on the utilization of one of the most  popular bandit algorithms that deal with stationary stochastic MABs, \ien, the Upper Confidence Bound algorithm \cite{auer2002finite}, due to its inherent simplicity, efficiency, and compatibility with resource-limited environments.
UCB maintains an estimate of the mean reward of each arm $\hat{\theta}_{a,t}$, and considers an upper bound  $U_{a,t}$ over the empirical mean:
\begin{equation}
U_{a,t}=\hat{\theta}_{a,t}+\sqrt{\dfrac{2\ln N_t}{N_{a,t}}},
\end{equation}
where $N_{a,t}$ and $N_t$ are the number observations of arm $a$  and the total number of observations up to time $t$ respectively. 
Here, $N_{a,t}$, $N_t$, $U_{a,t}$ and $\hat{\theta}_{a,t}$ are only updated whenever a new packet is observed.
At each time step, it selects the arm of the highest upper bound, \ien, $\argmax_{a\in\mathcal{K}} U_{a,t}$. 

In the following proposition, we present the expected number of observations that holds for any MAB algorithm and any sampling policy within this queue system.
This value significantly impacts system performance, as demonstrated in the sequel, given its direct impact on updating the MAB algorithm and subsequent convergence toward optimal arm selection.

\begin{prop}
\label{prop:nb_of_observ}
    Given a Geo/Geo/1 queue of infinite length, with arrival rate $\lambda$, and  service rate $\mu$.
    The expected number of served packets at time $T$ is:
    $$\mathbb{E}[N_T]=\min (\lambda,\mu)\times T.$$
\end{prop}

\begin{proof}
    If $\mu \leq \lambda$, then the expected number of served packets (observations) at time $T$ is $\mathbb{E}[N_T]=\mu \times T.$
    If $\mu \geq \lambda$, then $\lambda$ limits the number of packets that can be served. 
    The agent can serve only the packets that arrive, and the expected number of served packets remains $\lambda \times T$.
    Consequently, the expected number of served packets is $\mathbb{E}[N_T]=\min (\lambda,\mu)\times T.$
\end{proof}

\noindent
\textbf{UCB performance in QR-MAB. } 
In order to study the performance of UCB algorithm in remote control, we consider a $5$-armed bandit with Bernoulli rewards, such that $\theta_a \sim \mathcal{U}(0,1)$.
We run Algorithm \ref{alg:network-queue} for $T=5000$, and average over $1000$ Monte-Carlo simulations each of a different reward distribution. 
We compared the results of the most well-known standard sampling strategies, \ien, FIFO and LIFO.

\begin{figure}
    \centering
    
    \begin{subfigure}{0.35\textwidth}
        \includegraphics[width=\linewidth]{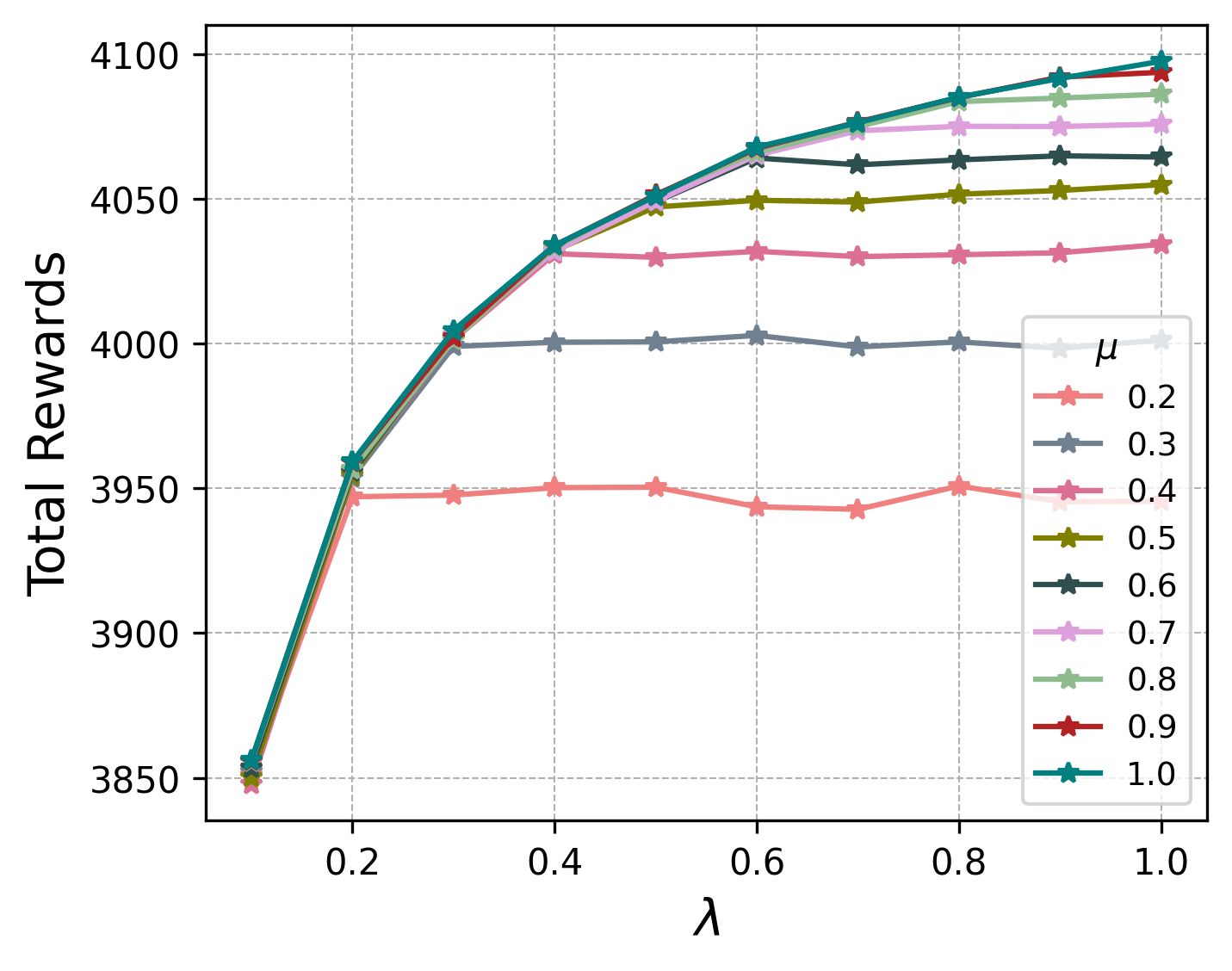}
        \caption{LIFO}
        \label{fig:lifo_lambda}
    \end{subfigure}
    \begin{subfigure}{0.35\textwidth}
        \includegraphics[width=\linewidth]{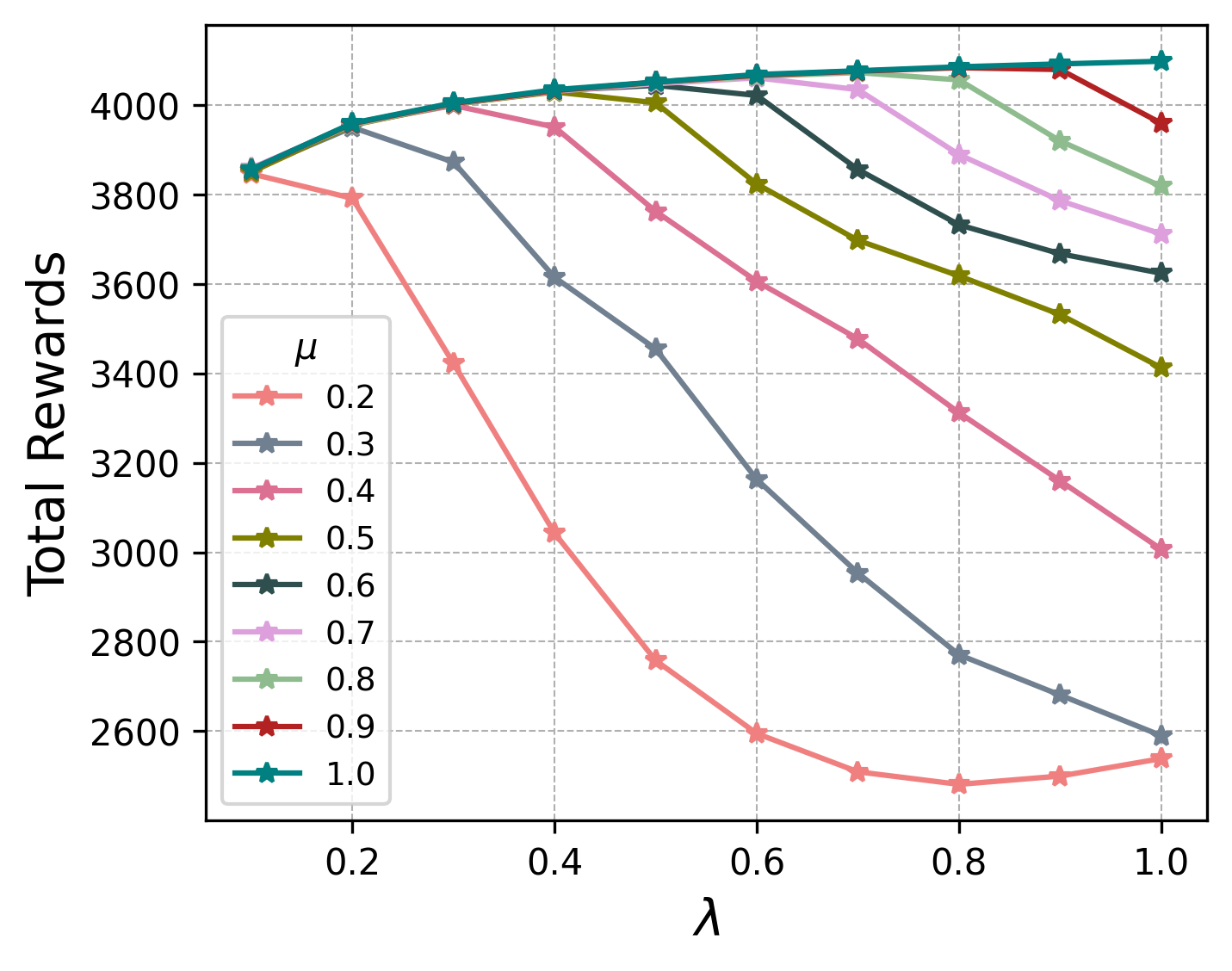}
        \caption{FIFO}
        \label{fig:fifo_lambda}
    \end{subfigure}
    \caption{Total reward of the UCB algorithm under LIFO and FIFO sampling strategies.}
    \label{fig:ucb_standardAlgo_flifo}
\end{figure}

Fig. \ref{fig:ucb_standardAlgo_flifo} shows the total reward gained by UCB at $T=5000$ with LIFO and FIFO sampling strategies as a function of $\lambda$ for various values of $\mu$.
As shown in Fig. \ref{fig:lifo_lambda}, with LIFO, the total rewards increase as a function of $\lambda$, as long as $\lambda \leq \mu$, since the number of observations increases. 
Then, as $\lambda > \mu$, the cumulative rewards remain constant, since the number of observations is fixed as shown in Proposition \ref{prop:nb_of_observ}.
Similar trends can be observed  for FIFO when $\lambda < \mu$. 
However, when $\lambda \geq \mu$, the sum of rewards decreases with $\lambda$, although the number of observations is fixed.
As we consider here a stationary environment, the age of information (age of packet) is not effective, due to the constant reward distribution over time.
To explain the behaviour of FIFO sampling strategy, we need to note that, since UCB is a deterministic policy (it selects the action of the highest UCB), in the absence of observations it selects the same action until a new observation is available, and hence an update occurs.
This results in a stack\footnote{To enhance readability, we will use the term ``stack" to refer to a set of consecutive packets belonging to the same arm in the queue.} of packets corresponding to the same arm in the queue, which grows in size if $\lambda\geq\mu$.
Consequently, in FIFO, the agent needs to observe all the packets of the stack sequentially before observing another arm.
This results in the updating of only one arm during this period, potentially leading to inaccurate estimations of the arms. 
In this case, the agent may end up with sub-optimal arms.
Conversely, this is less likely to occur with the LIFO policy, since the stack is broken with the last packet arriving to the queue and being observed by the agent.

Following LIFO sampling strategy, the agent is able to break a packet stack. 
However, it remains inherently sequential and the observation of several packets with the same arm could occur due to the absence of arrivals.


\section{Stochastic Biased Sampling Strategy}

We observed that employing standard queuing strategies such as FIFO and LIFO can cause the agent to process entire stacks sequentially, leading to the same arm being updated repeatedly, which prevents the observation of feedback from other arms.
This lack of randomness may lead to suboptimal performance, due to mis-estimating the mean rewards of other unobserved arms. 

To overcome this issue, a more dynamic approach is needed, to avoid sequential stack processing through random packet selection. 
For this sake, we propose the stochastic biased sampling strategy $\pi_{\delta u}(t)$, which samples the packet $(a,r)_s$ with probability $f_{\delta u}(s)$, where:
\begin{align}
    f_{\delta u}(s)&=\alpha \times \delta(s-c)+(1-\alpha)\times\mathcal{U}(L(t)),
\end{align}
where, $\alpha \in [0,1]$ is the randomness-control parameter, $\delta(s-c)$ is the Dirac function centered at $c$, which is an index of a packet in the queue and $\mathcal{U}(L(t))$ is the uniform distribution over the interval $[0,L(t)]$, \ien, the packets indices in the queue.   
In practice, $\pi_{\delta u}(t)$ samples the packet of index $c$ with probability $\alpha$, and samples a uniformly selected packet in the queue with a probability $1-\alpha$.
Such a sampling strategy introduces randomness, ensuring that all actions are observed and updated fairly.
In the following, we study the performance of this stochastic biased strategy.


\begin{figure}[t]
    \centering
    
    \includegraphics[width=0.75\linewidth]{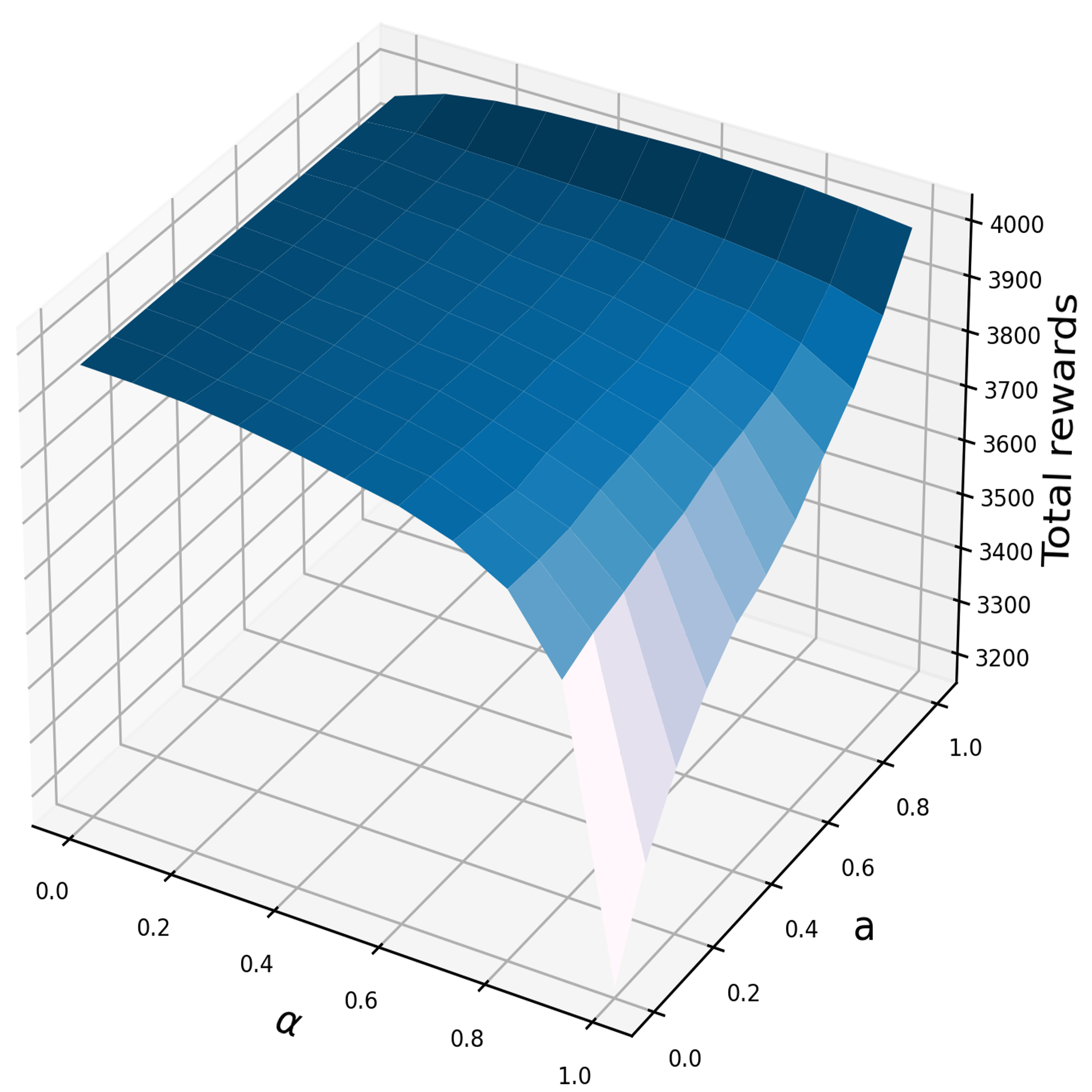}
        
    \caption{UCB total rewards with $\pi_{\delta u}$ with different values of $\alpha$ and $a$ at $\mu=0.3$ and $\lambda=0.6$}
    \label{fig:ucb_standardAlgo_DU}
\end{figure}

\begin{figure}[t]
    \centering
    
    \includegraphics[width=0.75\linewidth]{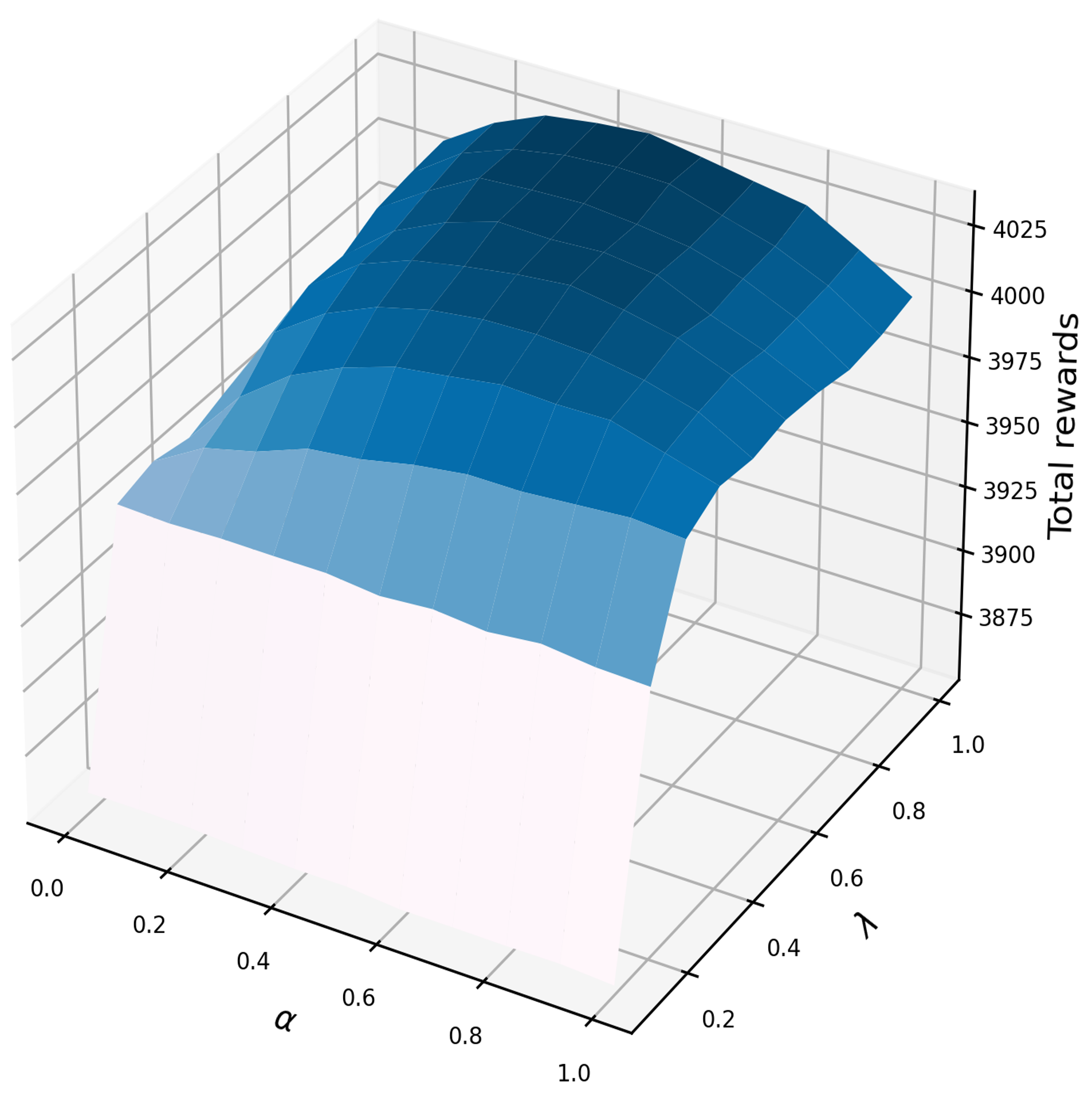}
        
    \caption{UCB total rewards following $\pi_{\delta u}$ sampling strategy with $a=1$, $\mu=0.3$ and different values of $\alpha$ and $\lambda$.}
    \label{fig:ucb_standardAlgo_DU-0.3}
\end{figure}

\section{Numerical Analysis}

We run the previous experiments with $\pi_{\delta u}$ and different values of $c$ and $\alpha$. 
As the queue length $L$ changes over time, the value of $c$ varies with the queue length. 
We consider $c(t)=a\times L(t)$ where $a\in [0,1]$. 
The value of $a$ controls Dirac's centering with respect to the queue length, \eg if $a=1$, then the Dirac is centered at the last packet arrived to the queue.

Fig. \ref{fig:ucb_standardAlgo_DU} presents the total rewards obtained by UCB while following the stochastic biased sampling strategy $\pi_{\delta u}$ with different values of $\alpha$ and $a$ (equivalently, $c$). 
As we are interested in the case when $\lambda > \mu$, we consider $\lambda=0.6$ and $\mu=0.3$.
We can notice the huge gap in performance between LIFO ($a=1, \alpha=1$) and FIFO ($a=0, \alpha=1$).
Also, it is evident that the maximum total reward is obtained with $a=1$, which corresponds to the case in which the last packet admitted to the queue is sampled with probability $\alpha$.

In Fig. \ref{fig:ucb_standardAlgo_DU-0.3}, the total rewards obtained by UCB with $a=1$ and $\mu=0.3$ are illustrated for different values of $\lambda$ and $\alpha$.
We can notice no difference in performance when $\lambda < \mu$, since as long as the service rate is greater than the arrival rate, the queue is stable, preventing the accumulation of a long backlog in the buffer, with subsequent no impact of $\alpha$.
On the other hand, we notice that the optimal value of $\alpha$ corresponding to the best performance varies with $\lambda$, however, it consistently converges around $0.5$, outperforming both LIFO ($\alpha=1$) and pure uniform sampling strategy ($\alpha=0$).
This implies that incorporating uniform sampling alongside prioritizing the most recently entered packet could be beneficial, serving to disrupt the formation of packet stacks in the queue.

\subsection{Regret performance of stochastic biased sampling strategy vs. literature}


\begin{figure}[t]
    \centering
    
    \includegraphics[width=\linewidth]{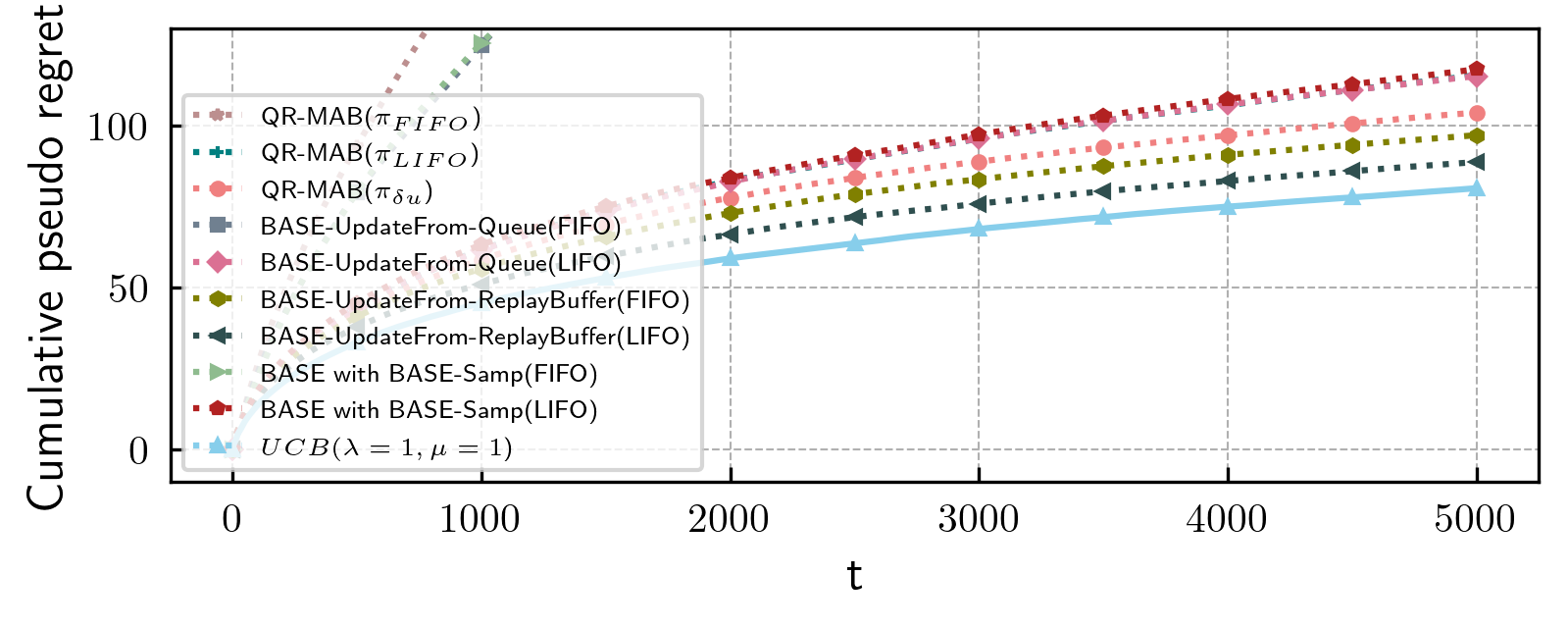}
        
    \caption{Cumulative pseudo regret of UCB for different policies with $\lambda=0.8$, $\mu=0.6$.}
    \label{fig:ucb_regret}
\end{figure}
\begin{figure}
    \centering
    
    \includegraphics[width=\linewidth]{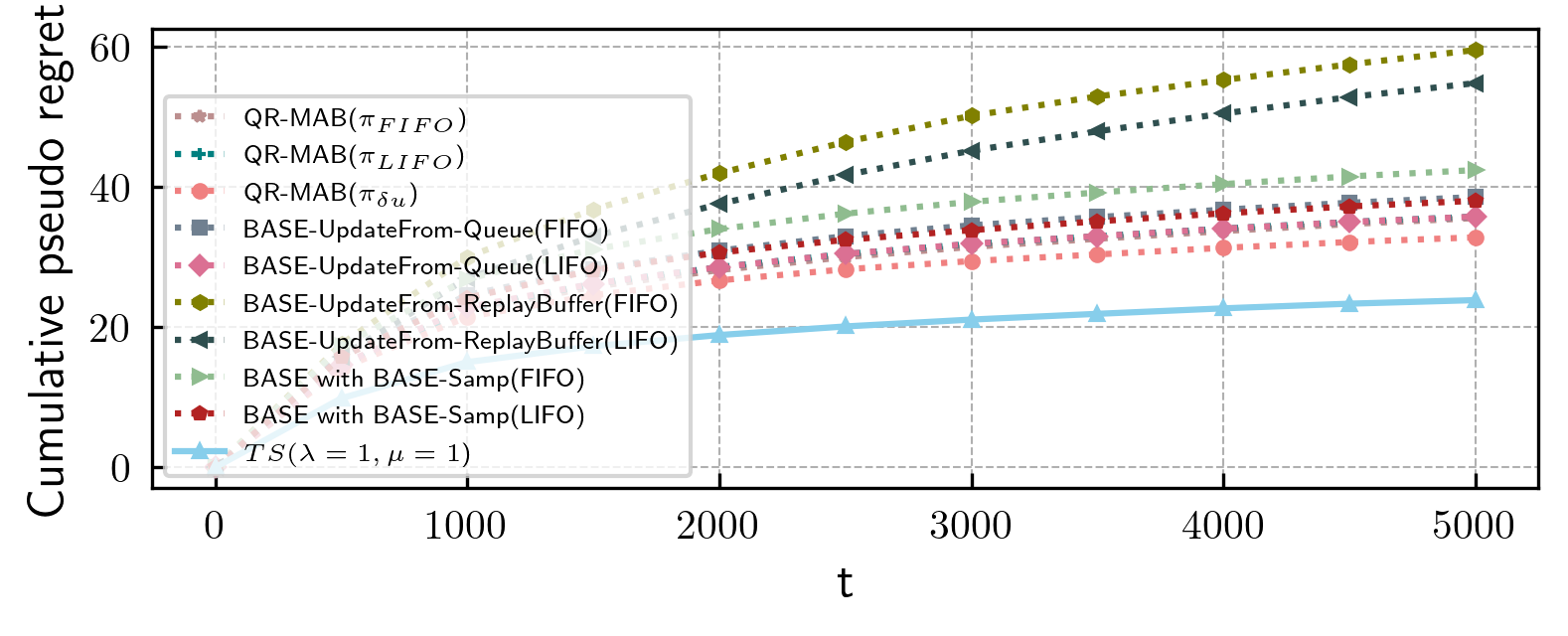}
        
    \caption{Cumulative pseudo regret of TS for different policies with $\lambda=0.8$, $\mu=0.6$.}
    \label{fig:ts_regret}
\end{figure}

We evaluate the performance of our stochastic biased sampling strategy against recent approaches from the literature that suit our settings.
In \cite{gupta2023remote}, the authors introduced three heuristics to mitigate the impact of delays in remote controlled MABs: BASE-UpdateFrom-Queue, BASE-UpdateFrom-ReplayBuffer, BASE with BASE-Samp.
Their common characteristic is the assumption of storage capabilities at the agent.
Their core concept is to maintain distinct queues for observed packets, one for each arm.
Following either FIFO or LIFO, the agent retrieves a packet from the network queue and places it in the corresponding arm queue. 
We hereby conduct a comparative analysis of their performance and compare them against QR-MAB with the stochastic biased sampling strategy.

Following the same experimental settings as for Fig. \ref{fig:ucb_standardAlgo_DU-0.3}, in Fig. \ref{fig:ucb_regret} we plot the cumulative pseudo-regret (Eq. \eqref{fig:ucb_regret}) for different policies with $\lambda=0.8 $ and $ \mu=0.6$.
We compare with the scenario without losses or delays (UCB($\lambda=1, \mu=1$)).
We can notice that QR-MAB with the stochastic biased sampling strategy (QR-MAB($\pi_{\delta u}$)) at $a=1$ and $\alpha=0.5$, outperforms all policies, with the exception of ``BASE-UpdateFrom-ReplayBuffer".
This is because the latter refrains from removing the packets from the arm queues, ensuring that the arms are consistently updated whenever at least one packet has been observed for that arm.
However, with this accumulation of reward data for all arms, the storage requirements can become substantial, potentially posing challenges in storage-constrained environments. 
Additionally, ``BASE-UpdateFrom-ReplayBuffer"  is associated with high-energy consumption as will be demonstrated in further experiments.
This introduces a notable trade-off between maximizing rewards and minimizing the energy and storage consumption. 

Furthermore, the authors in \cite{gupta2023remote} focused their analysis on Thompson Sampling algorithm \cite{kaufmann2012thompson}, showing that the latter does not exhibit any difference between LIFO and FIFO.
In Fig. \ref{fig:ts_regret}, we compare the cumulative pseudo-regret of the aforementioned policies to QR-MAB($\pi_{\delta u}$) with $a=1$ and $\alpha=0.5$ for the Thompson Sampling algorithm ($\mathcal{A}=$TS).
We can observe that our simple strategy achieves the lowest cumulative pseudo regret outperforming all other policies. 
This indicates that the randomness in updating the arms is more effective.
Besides, QR-MAB($\pi_{\delta u}$) is less storage and energy consuming. 

\subsection{Energy consumption vs. Total rewards}
As the energy consumption is critical in various applications, we analyze it for the aforementioned algorithms and compare it to their respective total rewards. 
For this sake, we make use of the codeCarbon \cite{codecarbon} package, which enables the estimation and tracking of the energy consumption associated to a running code.
We estimate the energy consumed by different policies with UCB algorithm, by running the previous experiments and averaging over $500$ Monte-Carlo simulations.
Let $\E_p$ be the energy consumed and $\R_p$ the total rewards acquired by a policy $p$ within the timeframe up to $T=5000$.
$\E_{\max}$ and $\R_{\max}$ respectively denote the energy consumed and the total rewards acquired by UCB when $\lambda=1$ and $\mu=1$.
At the same time, we denote by $\E_{\min}$ and $\R_{\min}$ respectively the energy consumption and the total reward with random arm selection for $T=5000$; in other words, no feedback is needed and no learning takes place in the latter case.
A policy $p$ could be any of the aforementioned policies, \ien, QR-MAB, ``BASE-UpdateFrom-Queue", ``BASE-UpdateFrom-ReplayBuffer", \etc
In order to evaluate the relative performance of the algorithms, we define two metrics: (i) reward loss indicator (RLI): $\text{RLI}(p)=1-(\R_p-\R_{\min})/(\R_{\max}-\R_{\min})$, and (ii) energy saving indicator (ESI): $\text{ESI}(p)=1-(\E_p-\E_{\min})/(\E_{\max}-\E_{\min}).$

\begin{figure}[t]
    \centering

    \includegraphics[width=0.85\linewidth,clip]{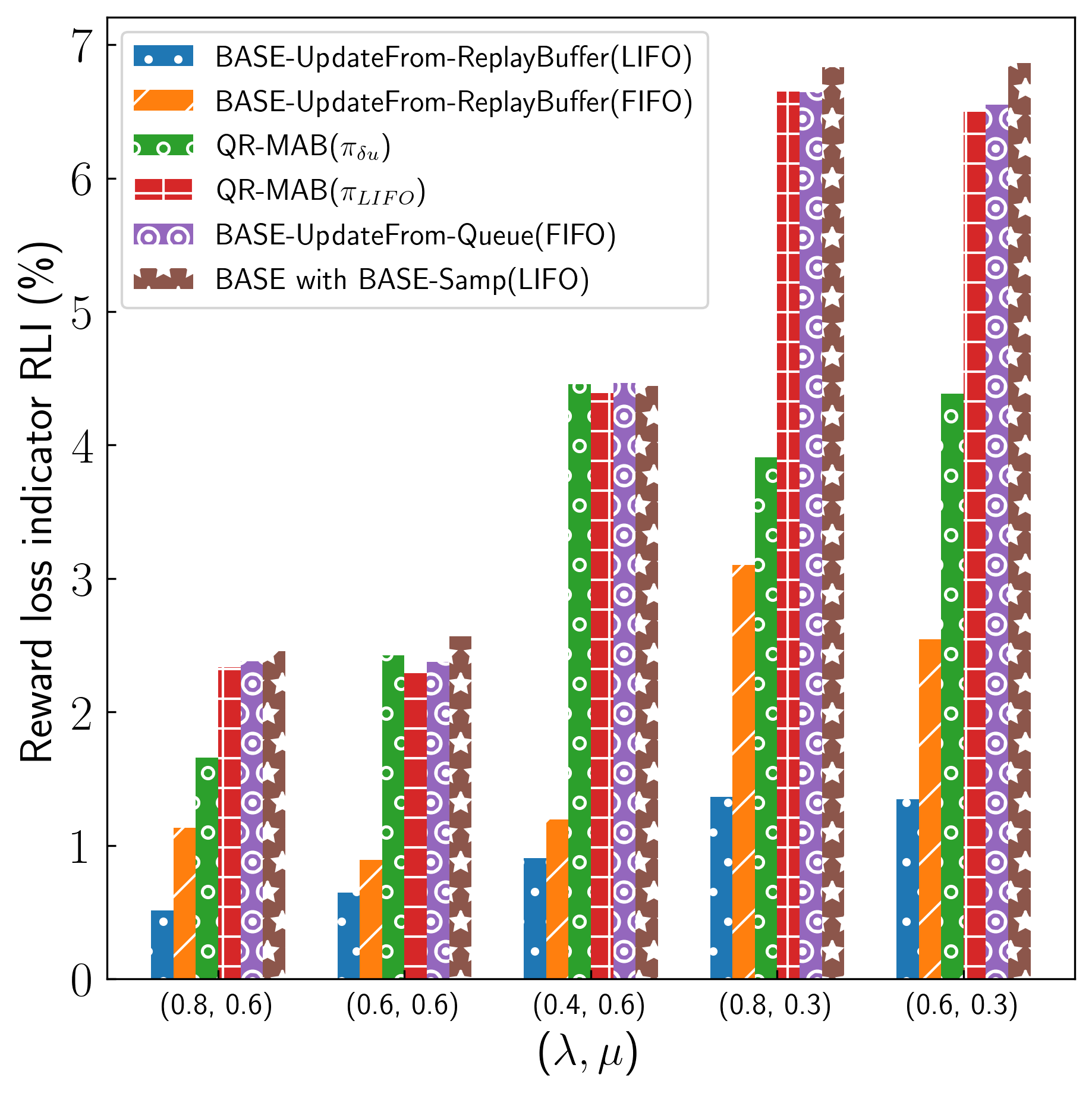}
        
    \caption{UCB reward loss rate of different algorithms at different $(\lambda,\mu)$ values.}
    \label{fig:RLR}
\end{figure}
\begin{figure}[t]
    \centering
    
    \includegraphics[width=0.85\linewidth, clip]{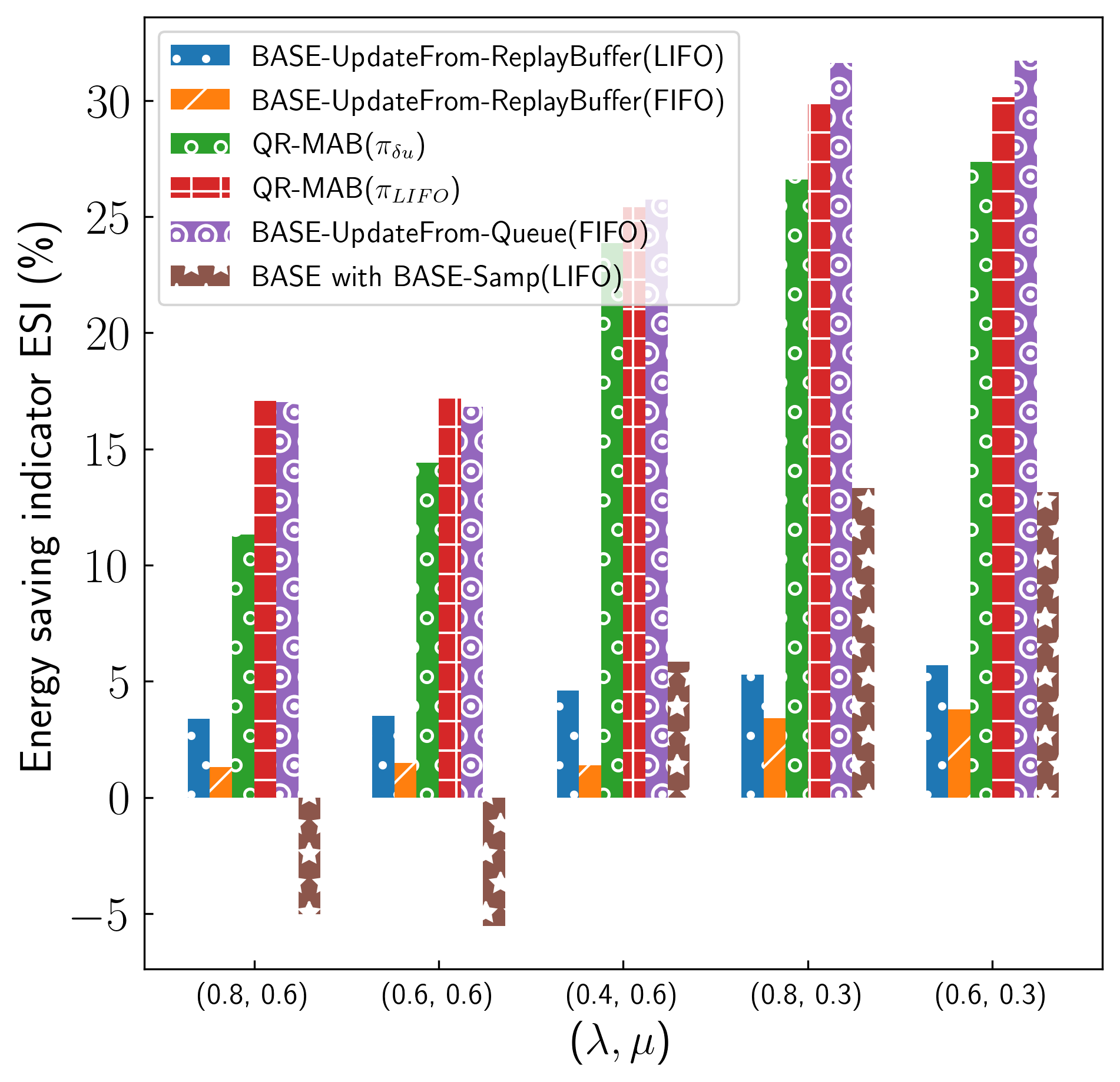}
        
    \caption{UCB energy saving rate of different algorithms at different $(\lambda,\mu)$ values.}
    \label{fig:ESR}
\end{figure}

Figs. \ref{fig:RLR} and \ref{fig:ESR} show the values of the RLI and ESI of different policies for different values of $(\lambda,\mu)$.
First, we can notice how the performance of QR-MAB with $\pi_{\delta u}$ sampling strategy improves when $\lambda > \mu$ in terms of reward loss rate.
Fig. \ref{fig:RLR} confirms the previous results of the cumulative regret in Fig. \ref{fig:ucb_regret}.
Although ``BASE-UpdateFrom-ReplayBuffer" exhibits a lower reward loss, not only its storage memory consumption is higher, but it is also high energy consuming, as shown in Fig. \ref{fig:ESR} (low ESR), whereas QR-MAB($\pi_{\delta u}$) is amongst the least energy-consuming algorithms. 
Taking the results for $(\lambda,\mu)=(0.8,0.3)$, we can notice that ``BASE-UpdateFrom-ReplayBuffer (LIFO)" achieves around $2\%$ lower RLI than QR-MAB ($\pi_{\delta u}$), while the latter saves around $20\%$ more energy. 
These results demonstrate a considerable trade-off between minimizing reward loss and maximizing energy saving, and the fact that QR-MAB($\pi_{\delta u}$) balances between achieving reasonable rewards and consuming energy. 
This potentially leads to more robust and sustainable performance. 


\section{Conclusion}
In this paper, we investigated the scenario of an agent controlling the actions of a remote actuator through MAB algorithms.
After each interaction with the environment, the actuator transmits the feedback to the agent over an unreliable channel.
Upon successful reception, the feedback is stored in a queue, to be then served by the agent, who employs a specific sampling strategy to select packets from the queue.
We showed that sequential sampling strategies, such as LIFO and FIFO, mitigate the performance of UCB algorithm, and proposed a novel stochastic strategy. 
The latter exhibits better regret performance than recent algorithms from the literature for TS algorithm. 
However, for the UCB, a notable trade-off exists between maximizing rewards and minimizing energy.
The new stochastic sampling strategy strikes a balance between optimizing rewards and minimizing energy consumption.
Future research will involve conducting theoretical regret analysis.
Additionally, this work could be further extended to address non-stationary environments, where the age of information is more significant.

\section*{Acknowledgements}
This work was funded by the French government under the France 2030 ANR program “PEPR Networks of the Future” (ref. 22-PEFT-0007).

\bibliographystyle{ieeetr}
\bibliography{biblio.bib}


\end{document}